\documentclass[12pt]{article}

\usepackage{fullpage}
\usepackage{latexsym}
\usepackage{epsfig}
\usepackage[usenames,dvipsnames,svgnames,table]{xcolor}
\usepackage{amsmath,amsthm,amssymb}
\usepackage[textsize=footnotesize]{todonotes}

\newtheorem{theorem}{Theorem}
\newtheorem{lemma}{Lemma}
\newtheorem{definition}{Definition}

\newcommand{\comment}[1]{{\color{red} #1}}

\title{The $(3,1)$-ordering for 4-connected planar triangulations}
\author{Therese Biedl
\thanks{David R.~Cheriton School of Computer Science, 
University of Waterloo, 
Waterloo, ON N2L 3G1, Canada, {\tt \{biedl,mderka\}@uwaterloo.ca}. Research supported by NSERC. The second author is supported by Vanier CGS. }
\and
Martin Derka
\addtocounter{footnote}{-1}
\footnotemark
}
\date{\today}

\begin{document}
\maketitle

\begin{abstract}
Canonical orderings of planar graphs have frequently been used in graph drawing and other graph algorithms. In this paper we introduce the notion of an
$(r,s)$-canonical order, which unifies many of the existing variants of
canonical orderings.  We then show that $(3,1)$-canonical ordering
for 4-connected triangulations always exist; to our knowledge
this variant of canonical ordering was not previously known.
We use it to give much simpler proofs of two previously known graph drawing 
results for 4-connected planar triangulations, namely, rectangular duals 
and rectangle-of-influence drawings.
\end{abstract}

\section{Background}
\label{se:intro}

A canonical ordering of a planar graph is a way of building the graph by iteratively attaching vertices to some ``basic graph'' (such as an edge) while preserving some connectivity invariant after each iteration. This concept was introduced in the late 1980's by de Fraysseix, Pach and Pollack~\cite{FPP90}. They used the canonical ordering to show that planar graphs can be drawn on a grid of size $(2n-4) \times (n-2)$. Subsequently, canonical orderings became one of the main tools in graph drawings, e.g.~for drawing graphs in grids of small dimensions (see e.g. \cite{FPP90,CN98}), rectangular duals~\cite{KH97}, and also graph algorithms such as encoding planar graphs \cite{HKL99} or finding $k$-disjoint trees in planar graphs~\cite{NRN97,NN00}. 

\paragraph{Our contribution}
There is now a number of variations of canonical orderings, depending on
the connectivity of the graph and whether it is triangulated or not.  (We
will review these below.)
In this paper, we show the existence yet another canonical ordering, this one  for planar 4-connected triangulations.  It is substantially different from the canonical ordering for such graphs that was presented by Kant and He~\cite{KH97}. 
We call this the $(3,1)$-canonical ordering.  More generally, we introduce
the concept of an $(r,s)$-canonical ordering, which (roughly speaking) means
that the partial graph must be $r$-connected and the rest-graph must be
$s$-connected; the existing canonical orders all fit into this framework.

We use the $(3,1)$-canonical ordering to provide alternate (and, in our opinion, significantly simpler) proofs of two previously known results about 4-connected planar triangulations: they have rectangular duals (Section~\ref{sec:rd}) and rectangle-of-influence drawings (Section~\ref{sec:ri-drawing}).

\section{Review of existing canonical orderings}

We assume that the reader is familiar with planar graphs (refer e.g.~to \cite{Die12}).  We use the
term {\em triangulation} for a maximal planar simple graph, i.e., a graph in which all faces
are triangles and which has $3n-6$ edges of which none is a multiple edge or a loop.  Such a graph has a unique planar embedding; we further
assume that one face has been fixed as the outer face.
We begin our review of canonical ordering with the
one for triangulations introduced by de Fraysseix et al.~\cite{FPP90}. 
We paraphrase their definition 
to the following one (which is easily shown to be equivalent):

\begin{definition}[Canonical ordering for triangulations~\cite{FPP90}]
\label{def:triangulated-co}
Let $G$ be a triangulation with outer face $u_1, u_2, u_3$. A vertex ordering
$v_1,\dots,v_n$ is called a \emph{canonical ordering} if 
\begin{itemize}
\item $v_1=u_1$, $v_2=u_2$, $v_n=u_3$,
\item For every $1 < k < n$
the subgraph $G_{k}$ of $G$ induced by vertices $v_1,v_2, \ldots, v_{k}$ is $2$-connected. 
\end{itemize}
\end{definition}

As we will see later, it will be convenient to define $V_k:=\{v_k\}$ and so $V_1\cup \dots\cup V_n$ becomes a partition
of the vertex set.  For any such partition and an index $k$, we use the notation $G_k$ for the subgraph induced by 
$V_1\cup \dots \cup V_k$ and we let the {\em complement}
$\overline{G_k}$ of $G_k$ be the subgraph induced by the vertices $V-(V_1\cup \dots \cup V_{k-1})$. 
Note that vertex set $V_k$ belongs to both $G_k$ and $\overline{G_k}$.
\footnote{Some references instead define $\overline{G_k}$ to be the subgraph induced by
$V-(V_1\cup \dots \cup V_k)$.  This complicates stating some of the conditions.}  

One can observe that in a canonical ordering for a triangulation, the complement $\overline{G_k}$ is a connected graph for all $k<n$. 
This holds because any vertex $v_k\neq u_1,u_2,u_3$ is not on the outer face and so there must exist some minimal $k'>k$ where $v_k$ is
not on the outer face of $G_{k'}$.  Due to the triangular faces, $v_k$ receives an edge to $v_{k'}$, and iterating the
argument, hence has a path within $\overline{G_k}$ that leads to $v_n$.

We note here, without giving details, that this canonical ordering
has been generalized to 3-connected planar graphs that are not necessarily
triangulated \cite{Kant96}, and also to non-planar
3-connected graphs (see \cite{Schmidt-ICALP14} and the references therein).

In 1997, Kant and He~\cite{KH97} showed that one can define a different canonical ordering for 4-connected triangulations, and used it to construct visibility representations of 4-connected planar graphs. 
Its definition, slightly paraphrased, is as follows:

\begin{definition}[Canonical ordering for 4-connected triangulations~\cite{KH97}]
\label{def:22-ordering}
Let $G$ be a 4-connected triangulation with outer face $u_1, u_2, u_3$. A vertex order
$v_1,\dots,v_n$ is called a \emph{canonical ordering for 4-connected triangulations} if 
\begin{itemize}
\item $v_1=u_1$, $v_2=u_2$, $v_n=u_3$,
\item For every $1 < k < n$, graphs $G_{k}$ and $\overline{G_k}$ are $2$-connected. 
\end{itemize}
\end{definition}

This canonical ordering was extended to a canonical ordering for all planar 4-connected graphs (not necessarily triangulated) by Nakano, Rahman and Nishizeki~\cite{NRN97}. 
Versions of a canonical order for 4-connected non-planar graphs are also known \cite{CLY05}.

Going one higher in connectivity,
Nagai and Nakano \cite{NN00} introduced a canonical ordering
for 5-connected triangulations.  Here, vertices are added in sets that
are sometimes more than a singleton.  We need a definition.
Let $G$ be a graph where all interior faces are triangles.  A {\em fan} 
of $G$ is a subset of
vertices $z_1,\dots,z_f$ that induces a path with
$\deg(z_i)=3$ for all $i=1,\dots,f$.  We will only apply this
concept if all vertices in the fan belong to the outer face of $G$.
Since interior faces are triangles, it follows that for all $z_i$
the third neighbor (i.e., the one not on the outer face) is the
same vertex.   See also Figure~\ref{fig:fan}(right).

\begin{definition}[Canonical ordering for 5-connected triangulations~\cite{NN00}]
Let $G$ be a 5-connected triangulation with outer face $u_1, u_2, u_3$. A partition of the vertices $V=V_1\cup \dots \cup V_L$ is called a \emph{canonical ordering for 5-connected triangulations} if 
\begin{itemize}
\item $V_1=\{u_1,u_2\}$,
\item $V_2$ consists of all neighbors of $u_1$ and $u_2$,
\item $V_L=\{u_3\}$,
\item $V_{L-1}$ consists of all neighbors of $u_3$,
\item For $2< k<  L-1$, vertex set $V_k$ is either a single vertex or a fan,
\item For every $2 < k < L$, graph $G_{k}$ is $3$-connected and
graph $\overline{G_{k}}$ is $2$-connected. 
\end{itemize}
\end{definition}

This canonical ordering was used to find 5 independent spanning trees in 5-connected triangulations~\cite{NN00}. To our knowledge, it has not been generalized to planar 5-connected (not necessarily triangulated) graphs, and not to non-planar 5-connected graphs either.    Since no planar graph is 6-connected, no canonical orderings for higher connectivity can exist for planar graphs.

Note that the three canonical orderings listed here are very similar, with the essence being the
connectivity that is required of the subgraphs and their complements.
In light of this, we 
aim to unify the three definitions with the following:

\begin{definition}[$(r,s)$-canonical ordering]
Let $G$ be a triangulation with outer-face $\{u_1,u_2,u_3\}$.
We say that 
a vertex partition $V_1 \cup \ldots \cup V_L$ 
is an \emph{$(r,s)$-canonical ordering} if 
\begin{itemize}
\item $u_1$ belongs to $V_1$ and $u_3$ belongs to $V_L$, and
\item for every $1 < k < L$, graph $G_k$ is $r$-connected and $\overline{G_k}$ is $s$-connected.
\end{itemize}
\end{definition}

Note that this definition is deliberately vague on the exact form that the vertex sets $V_k$ must have.
In particular, nothing prevents us (yet) from setting $L=1$ and $V_1=V$, which satisfies all conditions.
The existing canonical orderings restrict $V_k$ to be a singleton or, for $5$-connected triangulations, fans.
Thus the above definition should be viewed as a class of definitions, to be refined further by
stating restrictions on the vertex sets $V_k$.

Rephrasing the existing canonical orders in the above terms,
the canonical order for triangulations becomes a $(2,1)$-canonical ordering with only singletons, the one for 4-connected triangulations becomes a $(2,2)$-canonical ordering with only singletons, and the one for 5-connected triangulations becomes a $(3,2)$-canonical ordering with only singletons or fans. The reader will notice that the sum of the two numbers equals the connectivity of the graph.  Pushing this further, one may ask whether any $(r{+}s)$-connected
triangulation has an $(r,s)$-canonical ordering such that each $V_k$ has some simple form.
Note that we may assume that $r\geq s$, since a reversal of an $(r,s)$-canonical ordering gives an $(s,r)$-canonical ordering.
We study here $(3,1)$-canonical ordering for 4-connected triangulations, under the restriction that each $V_k$ is a singleton or a fan.
To our knowledge no such ordering was known before.

\section{$(3,1)$-canonical orderings}

We have already given the broad idea of a $(3,1)$-canonical ordering
earlier.  We re-state it here and give the specific
restrictions imposed  on the vertex sets.  See also`Figure~\ref{fig:fan}.

\begin{figure}[ht]
\hspace*{\fill}
\includegraphics[width=40mm,page=1]{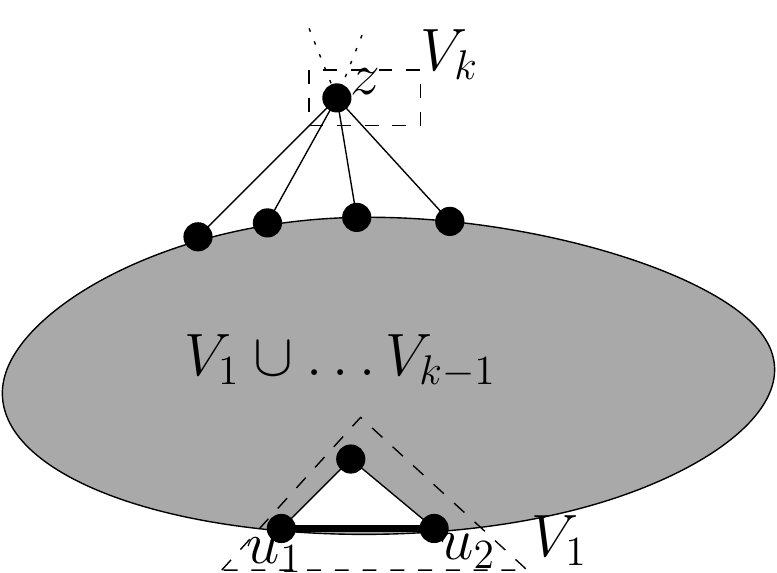}
\hspace*{\fill}
\includegraphics[width=40mm,page=2]{31-canonical.pdf}
\hspace*{\fill}
\caption{A singleton $V_k$ and a fan $V_k$ in a $(3,1)$-canonical
ordering.}
\label{fig:fan}
\end{figure}

\begin{definition}
Let $G$ be a 4-connected triangulation with outer-face $\{u_1,u_2,u_3\}$.
A {\em $(3,1)$-canonical order with singletons and fans} is a partition $V=V_1\cup \dots \cup V_L$
such that
\begin{itemize}
\item $V_1=\{u_1,u_2,z\}$, where $z$ is the third vertex of the
	interior face adjacent to $(u_1,u_2)$.
\item $V_L=\{u_3\}$.
\item For any $1<k<L$, set $V_k$ is either a singleton or a fan.
\item For any $1<k<L$, graph $G_k$ is $3$-connected 
	and $\overline{G_k}$ is connected.
\end{itemize}
\end{definition}

In what follows, we will omit the ``with singletons and fans'', as we
will not study any other version of $(3,1)$-canonical orderings.
Our main goal is to show that every 4-connected triangulation
has such a $(3,1)$-canonical ordering.  The proof of this proceeds
by induction, and we state the crucial lemma for the induction
step separately first.  We need a few definitions.

A plane graph is called a {\em triangulated disk} if every interior
face is a triangle and the outer-face is a simple cycle.
A triangulated disk is called {\em internally 4-connected} if its
outer-face has no chord, and every triangle is a face.
Observe that a triangle is an internally 4-connected triangulated disk,
and so is any 4-connected triangulation.  Also observe that a subgraph
of an internally 4-connected triangulated disk is again an internally
4-connected triangulated disk if and only if its outer-face is a simple
cycle that has no chord.

\begin{lemma}
\label{lem:find}
Let $G$ be an internally $4$-connected triangulated disk with $n\geq 4$.  
Let $(u_1,u_2)$ be an edge on the outer-face. 
Then there exists a vertex set $V'$ such that
\begin{itemize}
\item $V'$ contains only outer-face vertices, and none of $u_1,u_2$.
\item $G-V'$ is an internally $4$-connected triangulated disk.
\item $V'$ is a singleton or a fan.
\end{itemize}
\end{lemma}

\begin{proof}
\footnote{The proof is strongly inspired of the one 
for a $(3,2)$-canonical order in 5-connected graphs \cite{NN00}.
Since we demand
less on our $(3,1)$-canonical order, we can simplify the exposition
somewhat.}
Enumerate the outer face vertices of $G$ as $u_1=c_1,c_2,\dots,c_\ell=u_2$ in
clockwise order.  Define a {\em 2-leg} to be a path $c_i-x-c_j$ where
$i<j-1$ and $x$ is not on the outer-face.
Vertex $x$ is called a {\em 2-leg-center}.  We always have at least one
2-leg (namely, the one consisting of $u_1=c_1,u_2=c_\ell$ and their common 
neighbor at the interior face incident to $(u_1,u_2)$; this vertex is 
interior since $G$ has no chord and at least 4 vertices).

We say that a 2-leg-center $x$ {\em dominates} a 2-leg-center $y$ if vertex
$y$ is strictly inside the cycle $x-c_i-c_{i+1}-\dots-c_j-x$ formed by 
some 2-leg $\{c_i,x,c_j\}$ with center-vertex $x$.  See also
Figure~\ref{fig:2leg}(left).
The dominance-relationship 
is acyclic since any 2-leg with center-vertex $y$ must enclose strictly 
fewer faces than the 2-leg $\{c_i,x,c_j\}$.  Therefore we must have some
{\em minimal} 2-leg-centers, which are the ones that do not dominate any
other 2-leg-center.

By definition for any 2-leg $\{c_i,x,c_j\}$, we have $j\geq i+2$ and so 
there exists at least one vertex between $c_i$ and $c_j$
on the outer-face.  
We say that a 2-leg $\{c_i,x,c_j\}$ is {\em basic} if the vertices
$c_{i+1},\dots,c_{j-1}$ all have degree 3, and {\em complex} otherwise.
Note that if $\{c_i,x,c_j\}$ is basic, then $c_{i+1},\dots,c_{j-1}$ form
a fan and their common neighbor is $x$.

\begin{figure}[ht]
\hspace*{\fill}
\includegraphics[width=40mm,page=1]{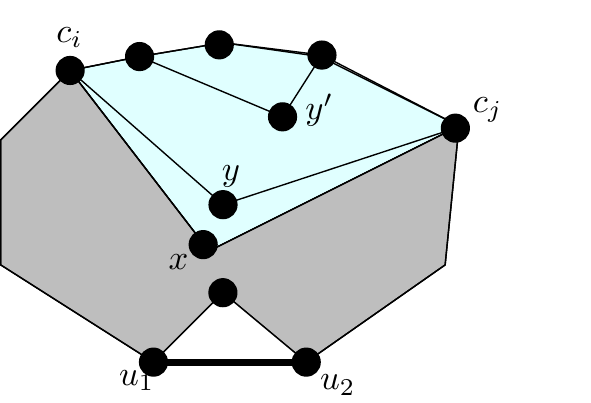}
\hspace*{\fill}
\includegraphics[width=40mm,page=6]{2leg.pdf}
\hspace*{\fill}
\includegraphics[width=40mm,page=4]{2leg.pdf}
\hspace*{\fill}
\caption{(Left) 2-leg center $x$ dominates both $y$ and $y'$.
(Middle)  If all 2-legs containing $x$ are basic, then we can remove a fan.
(Right) If $\{c_i,x,c_j\}$ is complex, then removing $c_{i+1}$ leaves
an internally 4-connected triangulated disk.}
\label{fig:2leg}
\end{figure}

Let $x$ be a minimal 2-leg center.
We have two cases:
\begin{itemize}
\item All 2-legs containing  $x$ are basic.

Let $i\geq 1$ be minimal and $j\leq \ell$ be maximal such that $x$ is adjacent
to $c_i$ and $c_j$.  
See also Figure~\ref{fig:2leg}(middle).
Since $x$ is a 2-leg-center, we have $i<j-1$.
By case assumption the 2-leg $\{c_i,x,c_j\}$ is basic, so 
$V'=\{c_{i+1},\dots,c_{j-1}\}$ is a fan. 
We verify that $G':=G-V'$ is an internally 4-connected triangulated disk:
\begin{itemize}
\item The outer-face of $G'$ consists of the one of $G$ plus $x$.
	By definition of a 2-center $x$ was not on the outer-face, so
	$G'$ is a triangulated disk.
\item Since $G$ had no chord, the only possible chord of $G'$ would be
	incident to vertex $x$.  But by choice of $i$ and $j$ the only
	neighbors of $x$ on the outer-face of $G'$ are $c_i$ and $c_j$.  So
	$G'$ has no chord.
\end{itemize}

\item Some 2-leg $\{c_i,x,c_j\}$ is complex.

We assume that $i$ has been chosen maximally, i.e., so 
that $\{c_{i+1},x,c_j\}$ is
either not a 2-leg or not complex.
We claim that in this case $V'=\{c_{i+1}\}$ is a suitable vertex set.

We first show that $c_{i+1}$ cannot be adjacent
to $x$.  Assume for contradiction that it is, then $\{c_i,x,c_{i+1}\}$ is
a triangle and hence a face.  If there were some $c_h$
with $i+1<h<j$ and $\deg(c_h)\geq 4$, then this would make
$\{c_{i+1},x,c_j\}$ a complex 2-leg, contradicting the choice of $i$.
So all of $c_{i+2},\dots,c_{j-1}$ (if any) have degree 3, and they
form a fan with common neighbor $x$.    In particular, edge $(c_{i+2},x)$
exists, which means triangle $\{c_{i+1},x,c_{i+2}\}$ is a face, forcing
$\deg(c_{i+1})=3$.  But then $\{c_i,x,c_j\}$ is basic, not complex.
This is a contradiction, so $x$ is not a neighbor of $c_{i+1}$.

Let $c_i=a_0,a_1,\dots,a_d,a_{d+1}=c_{i+2}$ be the neighbors of
$c_{i+1}$ in ccw order.  
See also Figure~\ref{fig:2leg}(right).
None of $a_1,\dots,a_d$ can be on the outer-face
of $G$, else $G$ would have a chord.  The outer-face of $G':=G-V'$
consists of $c_1,\dots,c_i,a_1,\dots,a_d,c_{i+1},\dots,c_\ell$, and so
this is a simple cycle and $G'$ is a triangulated disk.   Further, we
can show that it has no chord:
\begin{itemize}
\item If a chord of $G'$ connected two vertices in 
$c_1,\dots,c_i,c_{i+2},\dots,c_\ell$,
then it would also be a chord in $G$, which is excluded.
\item If a chord connected two non-consecutive vertices in
$c_i{=}a_0,\dots,a_{d+1}{=}c_{i+2}$, then in $G$ there would be an edge between two
non-consecutive neighbors of $c_{i+1}$, implying a triangle that is not a face.
\item If a chord connected some $a_s$, $1\leq s\leq d$, with some 
$c_h$, $i+2<h\leq j$, then $\{c_{i+1},a_s,c_h\}$ would be a 2-leg in $G$.
By minimality of $x$ hence $a_s=x$, but this contradicts that $c_{i+1}$ is
not adjacent to $x$.
\item If a chord connected some $a_s$, $1\leq s\leq d$, with some 
$c_h$, $1\leq h<i$ or $j<h\leq \ell$, then by $a_s\neq x$ it would have to cross $(c_i,x)$
or $(x,c_j)$, contradicting planarity.
\end{itemize}
So $G'$ is an internally 4-connected triangulated disk.  
\end{itemize}
Observe that in both cases
$V'\subseteq \{c_{i+1},\dots,c_{j-1}\}$
for some $1\leq i< j\leq \ell$, and so $V'$ does not contain $u_1$ or $u_2$
as desired.
\end{proof}

\begin{theorem}
Let $G$ be a 4-connected planar triangulation. 
Then $G$ has a $(3,1)$-canonical order.
\end{theorem}
\begin{proof}
We choose the vertex set in reverse order.    Let $\{u_1,u_2,u_3\}$ be
the outer-face and choose $V_L:=\{u_3\}$; this satisfies all conditions
since $u_3$ has at least 3 neighbors.  (We do not at this point know the
correct value of $L$, but simply assign indices backwards and shift
indices at the end so that the vertex sets are numbered $V_1,\dots,V_L$.)

Observe that $G-u_3$
is an internally 4-connected triangulated disk, because the neighbors
of $u_3$ form a simple cycle without chord (else there would be a separating
triangle at $u_3$).
Assume now some $V_{k+1},\dots,V_L$ have been chosen already such that
the remaining graph $G_k:=G-(V_{k+1}\cup \dots \cup V_L)$ is an
internally 4-connected triangulated disk with $(u_1,u_2)$ on the outer-face.
If $G_k$ has at least 4 vertices, then apply Lemma~\ref{lem:find} to find the next 
$V_k$.  
Graph $G_k-V_k$ is again internally 4-connected, so 
we can continue choosing vertex sets until only 3 vertices, including
$u_1$ and $u_2$, are left.
Since the graph is still internally 4-connected, these vertices must
be a triangle, and hence a face of $G$.  So setting $V_1$ to be the
three vertices of this triangle gives the desired ordering.

To observe that the required connectivity holds, note that any
internally 4-connected graph is 3-connected since it is a
triangulated disk without a chord.  To see that $\overline{G_k}$
is connected, it suffices to show that every vertex except $u_3$
has a neighbor in a later vertex set; the set of these edges
then forms a spanning tree in $\overline{G_k}$.  The argument for
this is nearly the same as for $(2,1)$-orderings.  Clearly 
each of $u_1,u_2$ are adjacent to $u_3$.  For any vertex $z\neq u_1,u_2,u_3$,
vertex $z$ is not on the outer face of $G$, and hence there must
exist some minimal $k'$ such that $z$ is on the outer face of
$G_{k'-1}$, but not on the outer face of $G_{k'}$.  Since faces
are triangles, this implies that $z$ is adjacent to some vertex
in $V_{k'}$.  By the above hence $\overline{G_k}$ is connected for
any $1<k<L$.
\end{proof}

The proofs of the above results are constructive and lead
to polynomial time algorithms.  With suitable data structures
to keep track of 2-leg-centers, it is not hard to see that 
a $(3,1)$-canonical ordering can be found in linear time; we
omit the details.

\section{Applications}
\label{se:appl}

In this section, we demonstrate two uses for the $(3,1)$-canonical
ordering in graph drawing.  Both results proved here were known before,
but in our opinion the $(3,1)$-canonical ordering significantly simplifies
the proof of these results.

\subsection{Rectangular duals}
\label{sec:rd}

A {\em rectangular dual drawing} (or {\em RD-drawing} for short) of a 
planar graph $G$ consists of
a set of interior-disjoint rectangles assigned to the vertices of $G$ 
in such a way that the union of the rectangles forms a rectangle without
holes, and the rectangles assigned to vertices $v$ and $w$ touch in a
non-zero-length line segment if and only if $(v,w)$ is an edge.
The following theorem has been proved repeatedly:

\begin{theorem}[\cite{Ung53,Tho84,KH97}]
Let $G$ be a $4$-connected triangulation, and let $e$
be an edge on the outer-face of $G$.  Then $G-e$ has a rectangular
dual.
\end{theorem}

Previous proofs on this result usually used the $(2,2)$-canonical
ordering (or some equivalent characterization, such as regular
edge labellings).  We give here a different proof using the
$(3,1)$-canonical ordering.

\begin{proof}
Let the outer-face be $\{u_1,u_2,u_3\}$, chosen such that $e=(u_1,u_2)$.
Find a $(3,1)$-canonical ordering $V_1\cup \dots\cup V_L$ of $G$.  We
now build the rectangular-dual drawing of $G-e$ by drawing $G_k-e$ for 
$k=1,\dots,L$.
By construction, $e=(u_1,u_2)$ is an edge on the outer-face of $G_k$, and
we can hence enumerate the outer-face of $G_k$ as $c^k_1,\dots,c^k_{\ell_k}$
with $c^k_1=u_1$ and $c^k_{\ell_k}=u_2$.  
We maintain the invariant
that in the RD-drawing of $G_k$, the rectangles of 
$c^k_1,\dots,c^k_{\ell_k}$ all attach at the top side of the bounding
box, in this order.

\begin{figure}[ht]
\hspace*{\fill}
\includegraphics[width=40mm,page=1]{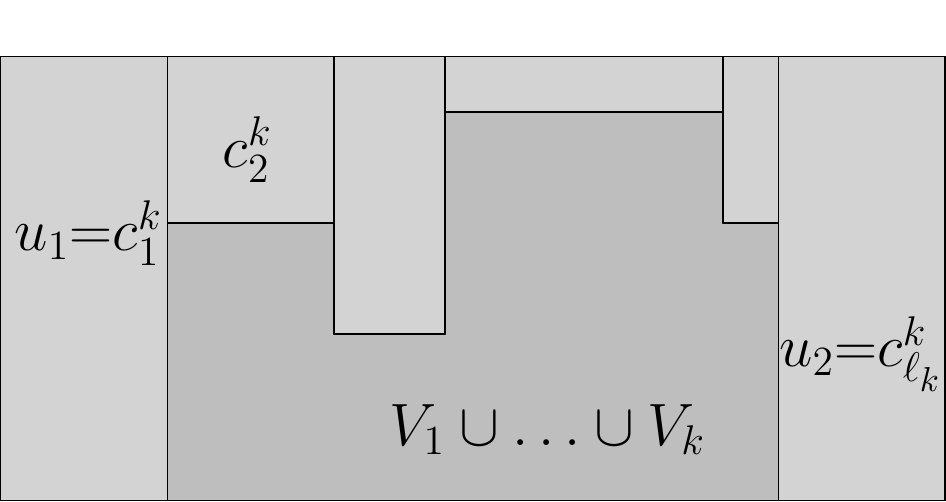}
\hspace*{\fill}
\includegraphics[width=40mm,page=2]{RD.pdf}
\hspace*{\fill}
\includegraphics[width=40mm,page=3]{RD.pdf}
\hspace*{\fill}
\caption{(Left) The invariant.  (Middle and right) 
Adding a singleton and a fan.  }
\label{fig:RD}`
\end{figure}

Such a drawing is easily created for $G_1-e$, since $G_1$ is a triangle and
so $G_1-e$ is a path $u_1-z-u_2$, where $z$ is the third vertex
of the interior face at $(u_1,u_2)$.
Now assume $G_k$ is drawn and consider adding either
a singleton or a fan $V_{k+1}$.  Let $a$ and $b$ be the smallest and 
largest index such that $c^k_a$ and $c^k_b$ are adjacent to a vertex in $V_{k+1}$.  

Extend all rectangles of $c^k_1,\dots,c^k_a$ and $c^k_b,\dots,c^k_{\ell_k}$ upward
by one unit.    This leaves a ``gap'' where the rectangles
of $c^k_{a+1},\dots,c^k_{b-1}$ ended. There is at least one such rectangle
since $b\geq a+2$ by properties of the $(3,1)$-canonical ordering (else $G_{k+1}$ would not be 3-connected).  
If $V_{k+1}$ is a singleton $z$, then we insert the rectangle for $z$ into
this gap.  If $V_{k+1}$ is a fan $\{z_1,\dots,z_f\}$, then $b=a+2$
and so the gap consists exactly of the top of $c^k_{a+1}$.  Split this
range into $f$ pieces and assign rectangles for 
$z_1,\dots,z_f$ in this place.  One easily verifies that this represents
all added edges as contacts and satisfies the invariant.
So we have the desired RD-drawing.
\end{proof}

\subsection{Rectangle-of-influence drawings}
\label{sec:ri-drawing}

A planar straight-line drawing of a graph is called a {\em (weak, closed)
rectangle-of-influence drawing} (or {\em RI-drawing} for short) if
for any edge $(u,v)$ the {\em rectangle $R(u,v)$ defined by $u,v$} is
{\em empty}, i.e., contains no other points of vertices of the graph.  (It
may contain parts of other edges.)  Here,  $R(u,v)$ is the minimum
axis-aligned rectangle that contains the points of $u$ and $v$; it 
degenerates into a line segment if $u$ or $v$ are on a horizontal or
vertical line.  The following result is known:

\begin{theorem}[\cite{BBM-GD99}]
\label{thm:ri}
Let $G$ be a $4$-connected triangulation and let $e$ be one
edge of the outer-face.  Then $G-e$ has a (weak, closed) 
rectangle-of-influence drawing.
\end{theorem}

We re-prove this result using the $(3,1)$-canonical ordering.  We note
here that the drawing created is exactly the same as in \cite{BBM-GD99};
the difference lies in that we can find the next vertex set to add much
more easily with the $(3,1)$-canonical ordering.

\begin{proof}
Let the outer-face be $\{u_1,u_2,u_3\}$, chosen such that $e=(u_1,u_2)$.
Find a $(3,1)$-canonical ordering $V_1\cup \dots\cup V_L$ of $G$.  We
now build the RI-drawing of $G-e$ by drawing $G_k-e$ for $k=1,\dots,L$.
By construction $e=(u_1,u_2)$ is an edge on the outer-face of $G_k$, and
we can hence enumerate the outer-face of $G_k$ as $c^k_1,\dots,c^k_{\ell_k}$
with $c^k_1=u_1$ and $c^k_{\ell_k}=u_2$.  
We maintain the invariant
that in the RI-drawing of $G_k$
$$x(c^k_1)<x(c^k_2)< \dots < x(c^k_{\ell_k}) \quad \mbox{and} \quad
y(c^k_1)>y(c^k_2)> \dots > y(c^k_{\ell_k}).$$
\begin{figure}[ht]
\hspace*{\fill}
\includegraphics[width=40mm,page=1]{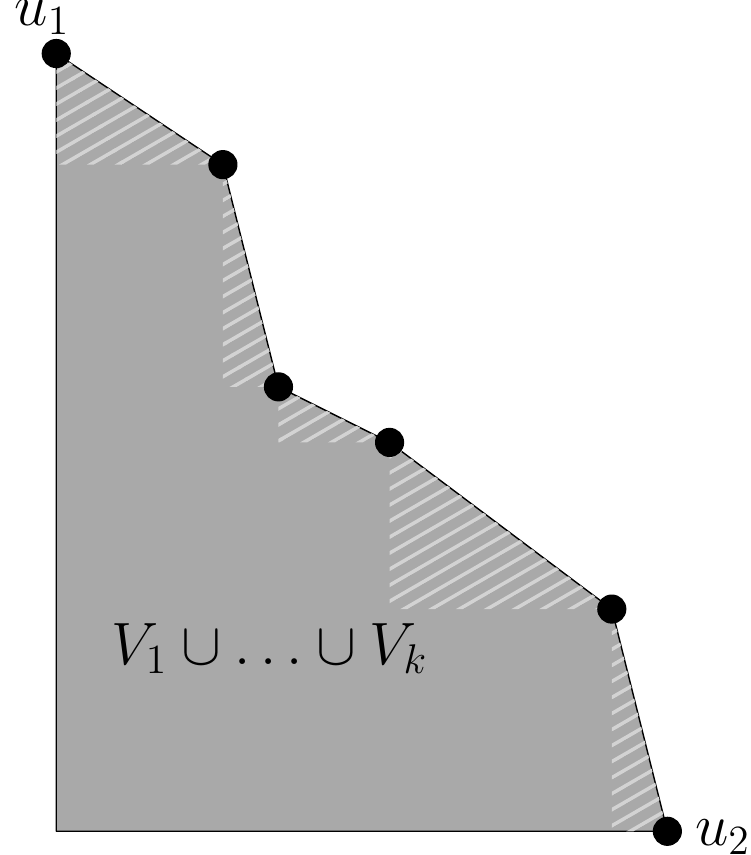}
\hspace*{\fill}
\includegraphics[width=40mm,page=2]{RI.pdf}
\hspace*{\fill}
\includegraphics[width=40mm,page=3]{RI.pdf}
\hspace*{\fill}
\caption{(Left) The invariant for RI-drawings.  Hatched regions contain no points due to the RI-drawing.  (Middle and right) Adding a singleton
and a fan.  The light gray region contains the new rectangles of influence.}
\label{fig:RI}`
\end{figure}

Such a drawing is easily created for $G_1-e$, since $G_1$ is a triangle and
so $G_1-e$ is a path $u_1-z-u_2$, where $z$ is the third vertex
of the interior face at $(u_1,u_2)$.
Now assume $G_k$ is drawn and consider adding either
a singleton or a fan $V_{k+1}$.  Let $a$ be the smallest and $b$ be the
largest index such that $c^k_a$ and $c^k_b$ are adjacent to a vertex in $V_{k+1}$.
By 3-connectivity of $G_{k+1}$ we have $b\geq a+2$.
If $V_{k+1}$ is a singleton $z$, then 
define 
$$x(z)= \frac{1}{2}\left(x(c^k_{b-1})+x(c^k_b)\right)$$
and $$y(z)= \frac{1}{2}\left(y(c^k_{a})+x(c^k_{a+1})\right).$$  See also
Figure~\ref{fig:RI}(middle).  By $a\leq b-2$ adding this new point satisfies
the invariant.
All rectangles $R(z,c^k_j)$ are empty for $a\leq j\leq b$, because they do not intersect
the drawing of $G_k$ except in rectangles $R(c^k_a,c^k_{a+1})$ and $R(c^k_{b-1},c^k_b)$.
So we have the desired RI-drawing.

If $V_{k+1}$ is a fan $\{z_1,\dots,z_f\}$, then $b=a+2$.
For $h=1,\dots,f$, define $$x(z_h)= \frac{h}{f+1}\left(x(c^k_{b-1})+x(c^k_b)\right)$$
and $$y(z_h)= \frac{f-h+1}{f+1}\left(y(c^k_{a})+x(c^k_{a+1})\right).$$  See also
Figure~\ref{fig:RI}(right).  By $a=b-2$ adding these new points satisfies the invariant.
All rectangles $R(z_h,c^k_j)$ are empty for $a\leq j\leq b$, because they do not intersect
the drawing of $G_k$ except in rectangles $R(c^k_a,c^k_{a+1})$ and $R(c^k_{b-1},c^k_b)$.
So we have the desired RI-drawing.
\end{proof}

\section{Conclusion}
\label{se:concl}

We showed the existence of new canonical order for $4$-connected triangulations. We used this canonical order to give simplified proofs of some previously known graph drawing results for 4-connected triangulations. 
Furthermore, we provided provided a brief survey of canonical  orderings for planar graphs and laid the groundwork for their further investigation.   Of particular interest to us are the following questions:
\begin{itemize}
\item Does every planar $c$-connected triangulation have an $(r,s)$-canonical ordering for all $r+s=c$ and reasonable restrictions
	on vertex sets  $V_k$? 
The missing case is a $(4,1)$-canonical ordering for 5-connected triangulations.

\item The $(r,s)$-canonical ordering definition naturally generalizes to planar graphs that are not necessarily triangulated.
	For the corresponding $(2,1)$-orderings \cite{Kant96} and $(2,2)$-orderings \cite{NRN97}
	it suffices to allow adding {\em chains}, i.e., induced paths.  Are there $(3,1)$-orderings, $(3,2)$-orderings and $(4,1)$-orderings 
	for 4-connected/5-connected planar graphs with some simple subgraphs as vertex sets $V_k$?  Likewise, exploration of
	$(r,s)$-canonical orders for non-planar graphs for $r+s\geq 5$ remains completely open.
\end{itemize}

\bibliographystyle{alpha}
\bibliography{therese}

\end{document}